\UseRawInputEncoding
\documentclass[conference]{IEEEtran}
\usepackage{lmodern}
\IEEEoverridecommandlockouts
\usepackage{cite}
\usepackage{amsmath,amssymb,amsfonts}
\usepackage{textcomp}
\usepackage{xcolor}
\def\BibTeX{{\rm B\kern-.05em{\sc i\kern-.025em b}\kern-.08em
    T\kern-.1667em\lower.7ex\hbox{E}\kern-.125emX}}

\long\def\comment#1{}
\usepackage{amsthm}
\usepackage[utf8x]{inputenc} % Ensure UTF-8 encoding
\usepackage{newtxtext}
\usepackage{newtxmath}

\usepackage{calc} % To reset the counter in the document after title page    
\usepackage{enumitem} % Includes lists

\usepackage{parskip}
\usepackage{hyperref}

\usepackage{algorithm}
\usepackage{algorithmic}
\usepackage{mathrsfs}  
\usepackage{booktabs}

\usepackage{times}
\usepackage{soul}
\usepackage{url}

\usepackage{setspace}
\usepackage{cite}
\usepackage{amsmath,amssymb,amsfonts}
\usepackage{graphicx}
\usepackage{textcomp}
\usepackage{xcolor}
\usepackage[detect-all]{siunitx}

\usepackage{xspace}
\usepackage{subcaption}
\usepackage{bm}
\usepackage{array}
\newcolumntype{P}[1]{>{\centering\arraybackslash}p{#1}}
\newcolumntype{M}[1]{>{\centering\arraybackslash}m{#1}}

\usepackage{algorithm,tabularx}
\usepackage{algorithmic}
\makeatletter
\newcommand{\multilines}[1]{%
	\begin{tabularx}{\dimexpr\linewidth-\ALG@thistlm}[t]{@{}X@{}}
		#1
	\end{tabularx}
}
\makeatother

\usepackage{mathtools}

\usepackage{amsthm}
\usepackage{multirow}
\usepackage{color}
\usepackage{epstopdf}
\usepackage{endnotes}
\usepackage{framed}

\usepackage{cool} % for partial derivative
\usepackage{enumerate} % for \begin{enumerate}[(a)], or [a)]
	\usepackage[subnum]{cases}
	\usepackage{multicol}
	\usepackage{tablefootnote}
	
	\DeclarePairedDelimiterX{\norm}[1]{\lVert}{\rVert}{#1}
	\DeclarePairedDelimiterX{\abs}[1]{\lvert}{\rvert}{#1}
	\DeclarePairedDelimiterX{\innProd}[1]{\langle}{\rangle}{#1}

	\newcommand{\defeq}{\vcentcolon=}
	\theoremstyle{plain}
	\newtheorem{theorem}{Theorem}
	\newtheorem{assumption}{Assumption}
	
	\newtheorem{lemma}{Lemma}
	
	\theoremstyle{definition}

\newcommand{\OurAlg}{\ensuremath{\textsc{GoodSpeed}}\xspace}

\begin{document}
\bstctlcite{IEEEexample:BSTcontrol}

\title{\OurAlg: Optimizing Fair Goodput with Adaptive Speculative Decoding in Distributed Edge Inference\\}
% \author{\IEEEauthorblockN{Phuong Tran, Tzu-Hao Liu, Long Tan Le, Tung-Anh Nguyen, Van Quan La, Eason Yu, Han Shu, \\Choong Seon Hong, Nguyen H. Tran}}

\author{\IEEEauthorblockN{Phuong Tran\textsuperscript{1},
 Tzu-Hao Liu\textsuperscript{1},  Long Tan Le\textsuperscript{1}, Tung-Anh Nguyen\textsuperscript{1}, Van Quan La\textsuperscript{1}, Eason Yu\textsuperscript{1}, Han Shu\textsuperscript{1} \\
 Choong Seon Hong\textsuperscript{2}, Nguyen H. Tran\textsuperscript{1}}
\IEEEauthorblockA{\textsuperscript{1}School of Computer Science,
The University of Sydney, Darlington, NSW 2006, Australia}
\IEEEauthorblockA{\textsuperscript{2}Department of Computer Science and Engineering, Kyung Hee University, Yongin-si, Gyeonggi-do 17104, Korea}}

\maketitle

\begin{abstract}
Large language models (LLMs) have revolutionized natural language processing, yet their high computational demands pose significant challenges for real-time inference, especially in multi-user and resource-constrained environments. Speculative decoding has emerged as a promising technique to accelerate LLM inference by using lightweight draft models to generate candidate tokens, which are subsequently verified by a larger, more accurate model. However, ensuring both high goodput—the effective rate of accepted tokens—and fairness across multiple draft servers cooperating with a central verification server remains an open challenge.
This paper introduces \OurAlg, a novel distributed inference framework that optimizes goodput through adaptive speculative decoding. \OurAlg employs a central verification server that coordinates a set of heterogeneous draft servers, each running a small language model to generate speculative tokens. To manage resource allocation effectively, \OurAlg incorporates a gradient scheduling algorithm that dynamically assigns token verification tasks, maximizing a logarithmic utility function to ensure proportional fairness across servers.
By processing speculative outputs from all draft servers in parallel, the framework enables efficient collaboration between the verification server and distributed draft generators, streamlining both latency and throughput. Through rigorous fluid sample path analysis, we show that \OurAlg converges to the optimal goodput allocation in steady-state conditions and maintains near-optimal performance with provably bounded error under dynamic workloads. These results demonstrate that \OurAlg provides a scalable, fair, and efficient solution for multi-server speculative decoding in distributed LLM inference systems.
\end{abstract}

\section{Introduction}
The emergence of agentic AI systems, autonomous agents powered by large language models (LLMs), e.g., Llama3-70B, capable of performing complex reasoning, planning, and communication, has introduced transformative capabilities for applications at the network edge. These agents facilitate tasks such as protocol adaptation, semantic data processing, and distributed decision-making. However, deploying such functions on resource-constrained edge hardware remains challenging, as multi-step inference and contextual reasoning typically exceed the computational and memory limitations of most edge devices. Small language models (SLMs), e.g., Qwen3-0.6B, while suitable for local deployment, often fall short in delivering the required accuracy for complex tasks. As a result, existing architectures rely heavily on offloading to either high-capacity servers or centralized cloud-based LLM services (e.g., OpenAI’s ChatGPT, Anthropic’s Claude, Microsoft’s AutoGen), where edge devices forward input data for remote inference.

The reliance on communication between edge devices and centralized cloud systems introduces several critical challenges. First, the performance of large language models (LLMs) is hindered by inference latency, which arises from both round-trip communication delays and the computationally intensive nature of LLM inference. These factors collectively increase response times and energy consumption, posing limitations for latency-sensitive interactive applications. Second, as the number of edge devices scales, centralized servers face growing computational and bandwidth demands, leading to scalability bottlenecks and elevated operational costs. These constraints complicate the deployment of LLMs that require low latency and high throughput responses. 

Recent developments in LLM serving systems offer potential solutions. DistServe improves resource efficiency by separating prefill and decoding phases across GPUs \cite{distserve2024}, while Sarathi-Serve addresses throughput-latency tradeoffs through chunked-prefills and stall-free scheduling \cite{sarathiserve2024}. The notion of smooth goodput has also been introduced as a unified metric to balance service-level objectives (SLOs) and performance \cite{smoothgoodput2024}. Among these, speculative decoding emerges as a promising approach, using smaller draft models to generate candidate token sequences verified in parallel by a more capable model, thus reducing latency. However, implementing speculative decoding in distributed edge systems is underexplored. 

To address the aforementioned challenges of distributed LLM inference, we propose \OurAlg, a framework for optimizing \underline{Good}put with Adaptive \underline{Spe}culative Decoding for Distributed \underline{Ed}ge Inference. \OurAlg equips edge servers with lightweight small language models (SLMs) that locally generate speculative tokens, which are subsequently verified by a central server running a full-scale large language model (LLM). Designed for deployment in resource-constrained and latency-sensitive environments, \OurAlg operates as a coordinated serving framework in which the central server dynamically manages verification and feedback across a heterogeneous pool of draft servers. This architecture enables scalable inference by parallelizing speculative token generation at the edge and leveraging centralized verification to minimize latency, while jointly optimizing goodput and fairness across the distributed system.
The key contributions are summarized as follows.
\begin{itemize}
    \item We propose \OurAlg, a distributed edge inference system that leverages speculative decoding to reduce latency. In \OurAlg, multiple draft servers use lightweight SLMs to autoregressively generate speculative token sequences. These sequences are then sent to a central verification server equipped with a large target model for validation. While each draft server generates tokens sequentially, \OurAlg enables parallel generation across these servers. Meanwhile, the verification server efficiently verifies batches of speculative tokens in parallel using GPU acceleration, significantly reducing the overall response time.
    \item  For efficiency and fairness, \OurAlg employs a gradient-based scheduling algorithm that dynamically determines the optimal number of speculative tokens to be generated by each draft server. The key innovation of this scheduling approach lies in its use of estimated token acceptance rates for the SLM at each draft server, which serve as critical inputs for estimating the expected goodput required by the scheduling process.
    \item We provide a rigorous theoretical analysis of \OurAlg. Using fluid model and sample path techniques, we prove that the sequence of token allocations generated by \OurAlg converges to a state that maximizes the system's goodput. 
    \item We conduct extensive experiments using state-of-the-art LLMs, including Qwen3-14B and Llama-3-70B, to validate our framework's performance. The results demonstrate that our dynamic scheduling algorithm achieves substantial gains in system-wide goodput and reduces end-to-end latency compared to baseline methods, confirming its  effectiveness.
\end{itemize}

\section{Background and Related Works}

\subsection{Speculative Decoding in LLM Inference}

\subsubsection{LLM Inference} The standard method for LLM inference is autoregressive decoding, where tokens are generated sequentially, with each new token conditioned on the full sequence of its predecessors~\cite{autoregressive_decoding}. Mathematically, given a sequence of tokens $s_1, s_2, \ldots, s_{j-1}$, an LLM model $M$ computes the conditional probability distribution for the next token $s_j$ as:
$$P_M(s_j \mid s_1, s_2, \ldots, s_{j-1}) = \text{softmax}(W \cdot h_{j-1})$$
where $h_{j-1}$ is the hidden state of the Transformer-based model $M$ at step $j-1$, and $W$ is the model's weight matrix mapping the hidden state to the vocabulary logits. The model $M$ then selects the next token (e.g., via greedy sampling, $s_j = \arg\max_{s} P_M(s \mid s_1, \ldots, s_{j-1})$) and appends it to the sequence.

In practice, this process involves two distinct computational phases: (1) a compute-bound \textit{prefill phase}, where the client's prompt is processed in a single, parallel forward pass across all tokens, dominated by large matrix multiplications and attention operations; and (2) a memory-bandwidth-bound \textit{decoding phase}, which generates the response one token at a time by repeatedly accessing a cached set of intermediate states. The sequential nature of the decoding phase represents a significant latency bottleneck in client-server applications \cite{LLMbottleneck}.

\subsubsection{Speculative Decoding}
Traditional decoding approaches, such as greedy or beam search, which generate tokens sequentially, requiring each one to be produced and verified before advancing to the next. This sequential dependency imposes significant latency, particularly in real-time or interactive scenarios. Speculative decoding (\textbf{SD}) mitigates this bottleneck by performing fast drafting and verification in parallel. It works as an efficient method for faster decoding without compromising quality. A lightweight but weaker draft model $M_q$ autoregressively generates $S$ candidate tokens by its sampling distribution $q_j(s) \defeq P_{M_q}(\cdot \mid, s_{1}, \dots, s_{j-1})$ 
$$s_{j} \sim q_j(s), \quad j=1,\dots,S.$$
A larger target model $M_p$ then computes the probabilities of candidate tokens using its sampling distribution $p_j(s) \defeq P_{M_p}(\cdot \mid s_{1}, \dots, s_{j-1}), j = 1, \ldots, S, $ \textit{in parallel}. 
%$$s_{j} \sim p(s_{j}), \quad i=1,2,\dots,S.$$
For each candidate $s_{j}$, draw a uniform random variable $r_j \sim U(0,1), j=1, \ldots, S$. The token $s_{j}$ is accepted if $r_j \leq  \frac{p_j(s_j)}{q_j(s_j)}$, ensuring the draft model's token is retained only if its probability aligns sufficiently with the target model's distribution; otherwise, it is rejected. Then, the number of accepted tokens is defined $m \defeq  \min ( \{j -1 \mid 1\leq j \leq S,  r_j > p_j(s_j)/q_j(s_j)\} \cup \{S\} )$. If $m < S$,  we sample the last token $s'$
using the normalized distribution of $\max(0, p_{m+1}(s) - q_{m+1}(s))$; otherwise, we sample $s'$ from $p_{S+1}(s)$. The accepted token output using SD will be $s_{\text{prefix}}, s_{1}, \ldots, s_m, s'$. 

From \cite{Yaniv}, the token acceptance probability is $$\alpha \defeq \mathbb{E}_{s \sim q(\cdot)} \left[ \min\left(1, {p(s)}/{q(s)}\right) \right].$$ 
Since the large model verifies candidate tokens in parallel, verification typically completes within the time of a single token generation, enabling a 2–3× speedup on T5-XXL compared with autoregressively decoding \cite{Yaniv}. The effectiveness of SD depends on the acceptance rate of drafted tokens, which is influenced by factors such as draft model quality and behavioral alignment between the draft and verification models \cite{xia2024unlockingefficiencylargelanguage}. A higher acceptance rate leads to greater speedup.

Recent advances in SD focus on improving the quality of speculative drafts, enabling self-speculative generation, and optimizing system-level efficiency. To enhance the quality and diversity of speculative drafts, SpecInfer \cite{SpecInfer} use multiple small models to generate diverse token trees at the cost of higher management overhead. The authors in \cite{LookAhead} reformulate decoding as a Jacobi iteration to enable parallel $n$-gram generation, but struggle with error propagation, where a single incorrect token can invalidate an entire draft. Other recent proposals focus on further improving the efficiency and verification process of these internal drafts~\cite{draft_verify, xia2025swift}. However, these self-speculation approaches require careful tuning of their architectural parameters to balance performance and accuracy.  

Despite these advances, existing SD approaches assume idealized execution environments, often ignoring constraints such as inter-model communication costs or hardware locality. Our work addresses this gap by \emph{modelling an SD-based multi-agent networked system, where draft and target models are distributed across physical nodes and interact under various input load datasets}.

\subsection{Distributed Edge Inference}

\subsubsection{Server-Side Inference Optimization}

A substantial body of work has addressed the server-side bottlenecks in LLM serving. Continuous batching has emerged as a central mechanism for improving throughput by processing inference requests as they arrive, rather than relying on static batching. This approach was introduced in Orca~\cite{Orca} and further advanced in vLLM~\cite{pagedattention}, which proposed PagedAttention to reduce memory fragmentation and improve memory reuse. Subsequent frameworks such as Sarathi-Serve~\cite{sarathiserve2024} and DeepSpeed-FastGen~\cite{DeepSpeed} incorporate techniques like dynamic input/output chunking to further increase system efficiency.
 
Systems including ShuffleInfer~\cite{TetriInfer}, Splitwise~\cite{SplitWise}, and DistServe~\cite{distserve2024} implement phase-specific scheduling to enhance parallelism and adapt to network conditions. At scale, model-parallel architectures such as HeteGen~\cite{HeteGen}, ExeGPT~\cite{ExeGPT}, and Helix~\cite{Helix} distribute model computation across GPUs or nodes to increase aggregate throughput, though often with a trade-off in per-request latency.

\subsubsection{Client-Side and On-Device Optimization}

Improving client-side performance is essential for edge deployments. Lightweight draft models can be made more efficient through post-training quantization and distillation techniques. Approaches such as GPTQ~\cite{gptq} and AWQ~\cite{awq} significantly reduce memory and compute requirements, making them well-suited for mobile or embedded devices.
Emerging work also explores self-speculative architectures that eliminate the need for server-side verification. Medusa~\cite{medusa} incorporates parallel decoding heads to perform both drafting and verification internally, while EAGLE~\cite{eagle} applies early-exit strategies to reduce inference latency without external assistance.

\subsection{Motivation}
Works like Ring Attention \cite{RingAttention} and TetriInfer \cite{TetriInfer} optimize system-level metrics (memory, latency) but do not ensure equitable goodput distribution across clients. Solutions such as Prompt Cache \cite{promptcache} and ExeGPT \cite{ExeGPT} assume static or predictable inputs, ignoring dynamic prompt variations. Methods like Helix \cite{Helix} and Mélange \cite{Melange} use intricate partitioning or allocation strategies, increasing computational overhead. Cloud deployments \cite{ServerlessLLM} focus on cost efficiency, often leaving server resources underutilized.  Moreover, traditional SD systems require a uniform draft length across batched requests, leading to inefficient decoding and limited scalability when handling prefixes with varying acceptance rates. 

To address these gaps, we introduce \OurAlg, a novel framework whose key contribution is a new edge-based distributed architecture for SD. This architecture leverages lightweight draft models on heterogeneous edge devices for local token generation, coupled with a central server for efficient verification and adaptive feedback, overcoming the latency and scalability limitations of centralized SD systems. \OurAlg further enhances this design with gradient scheduling to ensure fair resource allocation that facilitates efficient computation, and provides asymptotic optimality via fluid sample path analysis, validated by finite-interval experiments. 
\section{Problem Formulation}
\subsection{\OurAlg System Architecture}
\begin{figure*}[t]
    \centering
    \hypertarget{fig:system}{}
    \includegraphics[width=0.95\linewidth]{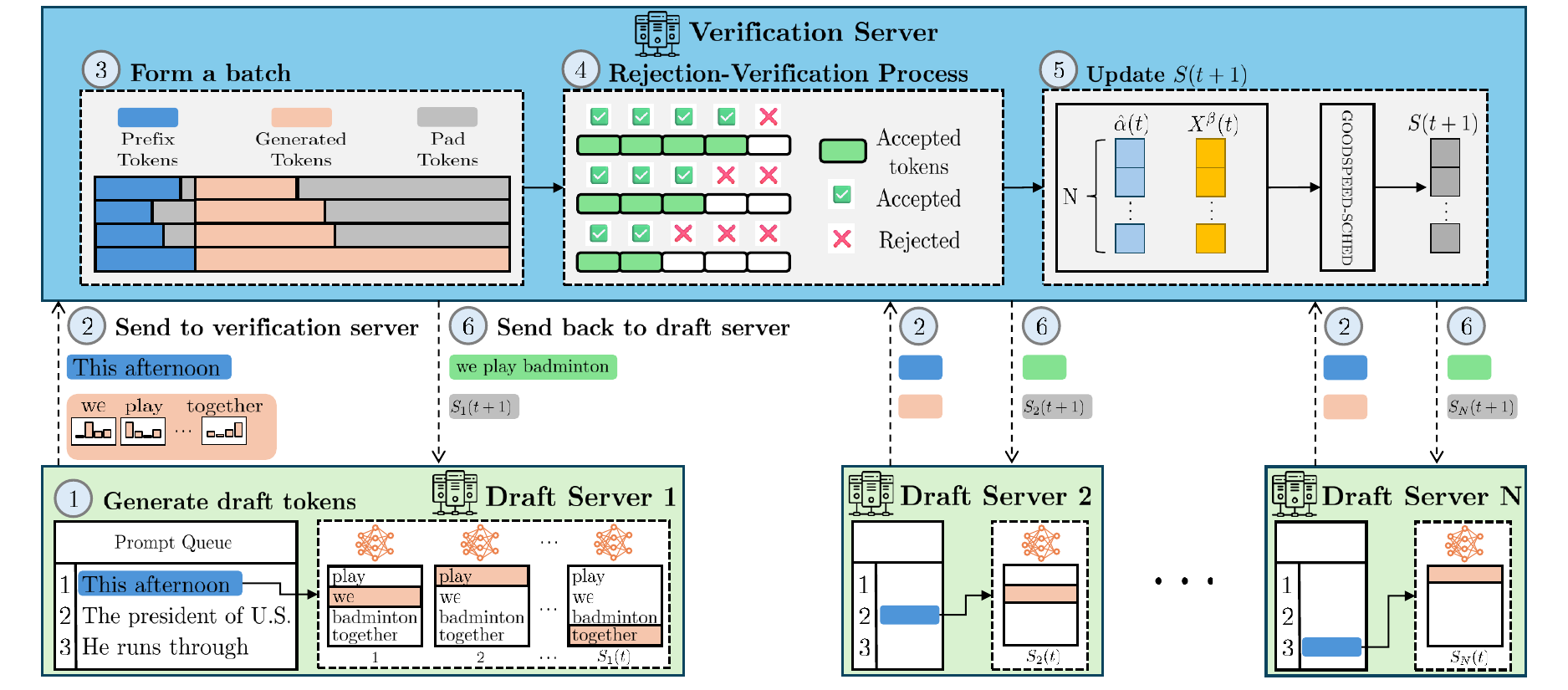}
    \caption{An overview of \OurAlg framework: \textcircled{{1}} Draft servers prepare draft tokens for current prompts; \textcircled{{2}} Draft servers send their draft tokens to the verification server; \textcircled{{3}} Draft sequences are batched; \textcircled{{4}} Rejection-verification process and \textcircled{{5}} \OurAlg's gradient scheduling for the allocation in the next round $t+1$; and \textcircled{{6}} The verification server sends optimal allocation to the draft servers.}
    \label{fig:system architecture}
\end{figure*}

Recent studies have extended SD to distributed settings. A common paradigm is edge–server collaboration, where lightweight draft models run on edge devices and a central server verifies proposed tokens. This structure supports heterogeneous devices while enabling the server to batch verification requests across clients, thereby improving utilization~\cite{li_sled_2025, spin}. To reduce communication overhead, network-efficient inference protocols have been proposed. For instance, DSSD~\cite{dssd2025} introduces a split verification mechanism that transmits only minimal accept/reject signals and required resample tokens, reducing bandwidth usage while preserving output fidelity.

%Building on these advancements, our \OurAlg framework adopts a similar edge-server collaboration paradigm but introduces a novel gradient-based scheduling algorithm to optimize verification slot allocations across clients, ensuring convergence to an optimal goodput distribution.

In this work, we consider a distributed edge inference system comprising $ N $ edge-deployed draft servers, each equipped with a lightweight draft model $ M_q^{(i)}, i = 1, \ldots, N $. These servers interact with a single verification server hosting a target (large) model $ M_p $ for verification, as illustrated in Figure~\ref{fig:system architecture}. The draft servers receive prompts from their respective end-users and engage in diverse real-time applications, including natural language translation, chatbot services, and code generation. This distributed edge inference setup allows for low-latency drafting close to end-users, minimizing data transfer overheads and enabling scalable real-time processing across geographically dispersed clients. The verification server, endowed with high-performance GPUs, serves as the central computational resource, validating the speculative tokens drafted by the small servers to ensure accuracy and efficiency. This validation employs a \emph{draft-then-verify} process inspired by SD. Lightweight models on small servers first generate candidate token sequences, proposing potential continuations (steps \hyperlink{fig:system}{\textcircled{1}} and \hyperlink{fig:system}{\textcircled{2}}). The verification server then verifies these drafts in parallel, accepting accurate prefixes that align with the target model’s distribution and rejecting mismatches. Corrections are sampled from an adjusted distribution, enabling faster inference while maintaining fidelity to the original output (steps \hyperlink{fig:system}{\textcircled{3}}, \hyperlink{fig:system}{\textcircled{4}} and \hyperlink{fig:system}{\textcircled{5}}). Each small server maintains a \emph{prefix} (e.g., a prompt) for its ongoing inference task, which is updated in every round by appending the accepted tokens and the correction token (if any) returned from the verification server after verification. This prefix serves as the conditioning context for generating the next set of draft tokens, ensuring that the inference process remains coherent and aligned with the target model's distribution across multiple rounds. The system operates in discrete time slots $ t = 0, 1, 2, \ldots $, and the interaction between the small servers and the verification server follows a structured flow: small servers submit draft tokens to the verification server, which manages a FIFO queue to process requests in the order of arrival, performs verification and rejection, and provides feedback to the small servers with updated allocations and accepted tokens (step \hyperlink{fig:system}{\textcircled{6}}). 
%The purpose of our framework is to design a scheduling algorithm that allocates the number of speculative tokens $ S_i(t) $ to each small server $ i $ at each round $ t $, such that the utility of the resulting goodput rates converges to an optimal allocation.
This distributed architecture underpins \OurAlg, designed to optimize inference performance across draft servers considering the verifying server's limited resources, particularly under non-stationary prompt conditions prevalent in dynamic operational environments.

\subsection{\OurAlg Optimization Problem}
\label{sec:primal}

\OurAlg is formulated by the following utility maximization problem:
\begin{align}
\label{primal}
\max_x \quad U(x) \defeq \sum\nolimits_{i=1}^N U_i(x_i) \quad \text{subject to} \quad x \in \mathcal{X}, 
\end{align}
where \(x_i\) denotes the goodput (i.e., expected number of accepted tokens per unit time) of draft server \(i\), \(x = (x_1, x_2, \dots, x_N) \in \mathbb{R}_{+}^N\) is the long-term average goodput vector, and \(\mathcal{X}\) denotes the achievable goodput region. This metric is computed as the limit of the empirical average $\bar{x}_i(T) = \frac{1}{T} \sum_{t=1}^T x_i(t)$ as the number of time steps $T$ approaches infinity, where $x_i(t)$ represents the realized goodput (accepted tokens) at time step $t$. The long-term average goodput captures the steady-state throughput achieved by each client under the framework's scheduling policy, accounting for variations in acceptance rates and resource allocation, and serves as the key performance metric optimized by the utility function $U(x)$ over the achievable region $\mathcal{X}$. In our paper, the utility function \(U_i: \mathbb{R}_{+} \to \mathbb{R}\) is continuously differentiable, strictly increasing, and strictly concave in each component (e.g., \(U_i(x_i) = \log x_i\) for proportional fairness \cite{Stolyar}).

To define the achievable goodput region \(\mathcal{X}\), consider a scheduling decision \(k = (S_1, S_2, \dots, S_N) \in \mathbb{Z}_{+}^N\), where \(S_i\) is the proposed length of draft tokens allocated to client \(i\). The expected goodput vector associated with \(k\) is \(\mu(k) = (\mu_1(k), \mu_2(k), \dots, \mu_N(k)) \in \mathbb{R}_{+}^N\)

\[
\mu_i(k) = \frac{1 - \alpha_i^{S_i + 1}}{1 - \alpha_i},
\]
representing the expected number of tokens generated for client \(i\) under speculative decoding \cite{Yaniv}, assuming a fixed acceptance rate \(0 < \alpha_i < 1\). This is the expected value of a geometric random variable with success probability \(1 - \alpha_i\), capped at \(S_i + 1\). We assume non-degeneracy bounds: there exist constants \(0 < \underline{\mu} \leq \bar{\mu} < \infty\) such that \(\underline{\mu} \leq \mu_i(k) \leq \bar{\mu}\) for all \(i = 1, \dots, N\) and all feasible \(k\), ensuring \(\mathcal{X}\) is bounded and non-degenerate.

The feasible set of \OurAlg's scheduling decisions is
\[
\mathcal{K} = \{ k = (S_1, S_2, \dots, S_N) \in \mathbb{Z}_{+}^N : \sum\nolimits_{i=1}^N S_i \leq C \},
\]
% where  \( C \) represents the verifying server's capacity, representing the maximum number of tokens the server can process.
where the constant \( C \) denotes the hardware budget of the verification server, representing the ideal number of tokens per forward pass to fully utilize both compute and memory bandwidth on modern GPUs \cite{li_adaserve_2025}. Sample $C$ values of H100 GPUs are given in Table~\ref{tab:skew_result}. It is selected through systematic hardware profiling to balance memory usage, latency, and throughput, while avoiding out-of-memory (OOM) errors at the verification server. 
%This choice follows adaptive LLM serving principles, as in AdaServe \cite{li_adaserve_2025}, to prevent resource overcommitment.}
%For each capacity state \(c\), let \(\phi_{c k} \geq 0\) denote the scheduling probability of choosing decision \(k \in K(c)\), with \(\sum_{k \in K(c)} \phi_{c k} = 1\). 
Let \(\phi(k) \geq 0\) denote the scheduling probability of choosing decision \(k \in \mathcal{K}\), with \(\sum_{k \in \mathcal{K}} \phi(k) = 1\). 
The goodput region \(\mathcal{X}\) is the convex hull of all achievable long-term average goodput vectors:
\begin{align}
\mathcal{X} =  \{ \bar{x} \mid \bar{x} =  \sum\nolimits_{k \in \mathcal{K} } \phi(k) \mu(k),
\sum\nolimits_{k \in \mathcal{K}} \phi(k) = 1 \}. 
\end{align}

The optimization problem \eqref{primal} admits a unique optimal solution \(x^* \in \mathcal{X}\) because \(U_i(x)\) is strictly concave and \(\mathcal{X}\) is convex and compact (guaranteeing that the maximum is attained).
However, the dynamic evolution of client prompts, which may transition abruptly between domains (e.g., casual dialogue to technical queries), introduces variability in the acceptance rate \( \alpha_i(t) \), the probability that a drafted token is accepted. Consequently, the problem \eqref{primal} involves dynamically allocating \( S_i(t) \) to optimize \( x_i(t) \) while adhering to the verification  server’s capacity \( C \). Therefore, directly solving this requires global knowledge of $ U(x) $, which depends on all possible $ \alpha_i(t) $ and $ C $, making it computationally infeasible in real-time, especially with non-stationary prompts.

\subsection{\OurAlg's Gradient Scheduling Algorithm}

The optimization problem \eqref{primal} defines a static benchmark for the optimal long-term goodput allocation $x^*$, under the assumption of fixed acceptance rates and averaged system behavior. However, distributed edge inference operates in a dynamic environment where acceptance rates $\alpha_i(t)$ fluctuate due to non-stationary prompts and task variations, and system capacity may vary over time. To adapt to these dynamics, \OurAlg employs a gradient-based scheduling algorithm that performs online allocation at each time step t. This algorithm adjusts the number of proposed tokens for each draft server for the next time slot $t+1$, using real-time system estimates. By doing so, it ensures adaptive resource allocation and asymptotically converges to the optimal allocation $x^*$, maintaining both efficiency and fairness in dynamic conditions.

Central to this is the use of smoothed estimates for both acceptance rates and goodput, which provide stable representations of each client's performance while incorporating recent observations. As described in \ref{algo:goodspeed}, at each time step \(t\), the framework maintains a vector of smoothed acceptance rate estimates \(\hat{\alpha}(t) = (\hat{\alpha}_1(t), \hat{\alpha}_2(t), \dots, \hat{\alpha}_N(t)) \in (0,1)^N\), where \(\hat{\alpha}_i(t)\) updates the probability that a drafted token of server $i$ is accepted by the target model, after the verification process (line~\ref{algo:update}). These estimates are updated using exponential smoothing by the verification server in parallel as follows:
\begin{align}
\label{updatealpha}
\hat{\alpha}_i(t) = (1 - \eta) \hat{\alpha}_i(t-1) + \eta \cdot \frac{1}{S_i(t)} \sum_{j=1}^{S_i(t)} \min\left(1, \frac{p_j(s_{j})}{q_{i,j}(s_{j})}\right)
\end{align}
% \begin{align}
% \label{updatealpha}
% \hat{\alpha}_i(t) = (1 - \eta) \hat{\alpha}_i(t-1) + \eta \mathbb{E}_{s \in S_i(t)} \left[ \min\left(1, {p(s)}/{q_i(s)}\right) \right]
% \end{align}
where the sum computes the empirical acceptance rate from the verification outcomes at time \(t\), with $p_j$ and $q_{i,j}$ being the target and draft server's model $i$ probabilities on token $s_j$, and $\eta \in (0,1)$ is a fixed smoothing parameter that controls the adaptation to recent observations. We proposed this smoothing estimate $\hat{\alpha}_i(t)$ to mitigate short-term noise from prompt variations, ensuring that \(\hat{\alpha}(t)\) converges to a stable value reflective of the system's ergodic behaviour under some assumptions (outlined in section \ref{sec:con_analysis}).

Building on the acceptance estimates, the framework also maintains a vector of smoothed goodput estimates \(X^\beta(t) = (X^\beta_1(t), X^\beta_2(t), \dots, X^\beta_N(t)) \in \mathbb{R}_{+}^N\), where \(X^\beta_i(t)\) approximates the expected number of accepted tokens per unit time for client \(i\) based on historical data (line~\ref{alg:update_X}). These estimates are updated using exponential smoothing by the verification server:
\begin{align}
\label{updateX}
X^\beta_i(t) = (1 - \beta) X^\beta_i(t-1) + \beta \, x_i(t),
\end{align}
where \(x_i(t)\) is the realized goodput for client \(i\) at time \(t\) (the number of accepted tokens plus one correction from verification \cite{leviathan2023fast}), and $\beta \in (0,1)$ is a fixed smoothing parameter. The smoothing for \(X(t)\) further stabilizes the system against fluctuations. Building on these estimates, the gradient scheduling problem optimizes the allocation of draft token slots \(S(t+1) = (S_1(t+1), S_2(t+1), \dots, S_N(t+1)) \in \mathbb{Z}_{+}^N\) at each time step \(t\) by maximizing the projection of the estimated goodput onto the gradient of the utility function (line~\ref{alg:goodspeed_sched}). This approach ensures that each allocation step favors draft servers with the highest marginal utility gain, \emph{promoting both efficiency and fairness while adapting to the current state} encoded in \(X(t)\) \cite{Stolyar}.
Mathematically, the gradient scheduling problem solved by the verification server at each time $t$ is formulated as follows.

% \OurAlg-\textsc{sched}:
\protect\hypertarget{goodspeedsched}{\OurAlg-\textsc{sched}:}
\begin{equation}
\label{eq:goodspeed_sched}
\begin{aligned}
\quad \max_{S(t+1)} \quad & \sum\nolimits_{i=1}^{N} \nabla U_i\big(X^\beta_i(t)\big) \, \hat{x}_i(t+1) \\
\quad \text{s.t.} \quad 
& \sum\nolimits_{i=1}^{N} S_i(t+1) \leq C, \quad S_i(t+1) \in \mathbb{Z}_{+}, \forall i.
\end{aligned}
\end{equation}
%  \[
% \max_{S(t+1)} \quad \sum\nolimits_{i=1}^N \nabla U_i (X(t)) \, \hat{x}_i (t+1)
% \]
% subject to
% \[
% \sum_{i=1}^N S_i(t+1) \leq C, \quad S_i(t+1) \in \mathbb{Z}_{+} \quad \forall i,
% \]
%where \(\nabla U(X(t)) = \left( \frac{\partial U_1}{\partial X_1(t)}, \dots, \frac{\partial U_N}{\partial X_N(t)} \right) \in \mathbb{R}_{+}^N\), and \(\hat{x}(t+1) = (\hat{x}_1(t+1), \dots, \hat{x}_N(t+1)) \in \mathbb{R}_{+}^N\) is the estimated goodput vector with

Here, $\hat{x}_i(t+1)$ is the estimated goodput of the draft server $i$ at timestep $t+1$, computed by the following: % at time slot $t+1$
\[
\hat{x}_i(t+1) = \frac{1 - \hat{\alpha}_i(t)^{S_i(t+1) + 1}}{1 - \hat{\alpha}_i(t)}, 
\]
where $\hat{\alpha}_i(t)$ is the estimated acceptance rate obtained by \eqref{updatealpha}.
%ßusing the smoothed acceptance rate \(\hat{\alpha}_i(t)\).

% The objective separates as \(\sum_{i=1}^N r_i(t) \cdot \hat{x}_i(t+1)\), where \(r_i(t) = \frac{\partial U_i}{\partial X_i(t)} > 0\). The function \(\hat{x}_i(t+1)\) is strictly increasing and concave in \(S_i(t+1)\), with marginal gain
% \[
% \Delta_i(s) = r_i(t) \cdot \alpha_i(t)^{s+1},
% \]
% positive and decreasing in \(s\).

% The optimal \(S(t+1)\) is solved via greedy marginal allocation:
% 1. Initialize \(S_i(t+1) = 0\) for all \(i\).
% 2. For each of \(C(t)\) slots: Compute \(\Delta_i(S_i(t+1))\); select \(i^* = \arg\max_i \Delta_i(S_i(t+1))\); increment \(S_{i^*}(t+1)\).

% This is optimal due to decreasing marginals and separability.

%----- Use Gurobi------

%\textcolor{blue}{We implement the \OurAlg algorithm with with package Gurobi}

\begin{algorithm}[t]
    \caption{\OurAlg Scheduling Algorithm}
    \label{algo:goodspeed}
    \begin{algorithmic}[1]
   % \STATE \textbf{Given:} $C$, $U_i : \mathbb{R}_+ \rightarrow \mathbb{R}$ for each draft server $i$
    \STATE \textbf{Initialize:} $\beta, \eta, S_i(0), X_i(0), \hat{\alpha}_i(0), $ $\forall i = 1, \dots, N$
    \FOR{each timestep $t$}
    \FOR{each \emph{draft server}, $i = 1, \dots, N$ parallelly}
        \STATE Sample $S_i(t)$ draft tokens from $M_q^{(i)}$ (\hyperlink{fig:system}{step \textcircled{1}}):
        \FOR{$j=1$ to $S_i(t)$}
            \STATE $q_{i,j}(s) \leftarrow P_{M_q^{(i)}}(\cdot  \mid {\text{prefix}_i},s_{i,1},...,s_{i,j-1})$
            \STATE $s_{i,j}(t) \sim q_{i,j}(s)$ 
            %\STATE $\{q_{i}(s)\}_{j=1}^{S_i}$
        \ENDFOR
        \STATE $s_i(t) \leftarrow (\text{prefix}_i(t), s_{i,1}(t),...s_{i,S_i(t)}(t))$
        \STATE Send $s_i(t)$, $q_{i,j}$ to the verification server (\hyperlink{fig:system}{step \textcircled{2}})
    \ENDFOR
    \STATE \emph{Verification server:} 
    \STATE Does Batching (\hyperlink{fig:system}{step \textcircled{3}}) and Verification (\hyperlink{fig:system}{step \textcircled{4}})
    \STATE \label{algo:update} Computes $x_i(t)$, updates $\hat{\alpha}_i(t)$ using \eqref{updatealpha}, and updates  $X_i(t)$ using \eqref{updateX}. \label{alg:update_X}
    \STATE Solves the \hyperlink{goodspeedsched}{\OurAlg-\textsc{sched}}  problem \eqref{eq:goodspeed_sched} to obtain $S(t+1)$ (\hyperlink{fig:system}{step \textcircled{5}}) \label{alg:goodspeed_sched}
    \STATE Sends $S(t+1)$ to draft servers (\hyperlink{fig:system}{step \textcircled{6}}). 
    \ENDFOR
    \end{algorithmic}
\end{algorithm}

    % \STATE Update $\hat{\alpha}_i(t)$ using \eqref{updatealpha} \label{alg:update_alpha}
    % \STATE Update $X_i(t)$ using \eqref{updateX} \label{alg:update_X}

% \STATE \hspace{0.4cm} $S(t+1) \leftarrow \max \sum_{i=1}^N \nabla U_i(X_i(t)) \cdot \left( \frac{1 - \hat{\alpha}_i(t)^{S_i(t + 1) + 1}}{1 - \hat{\alpha}_i(t)} \right)$
    % \STATE \hspace{0.4cm} \text{s.t.} $\sum_{i=1}^N S_i \leq C,\quad S_i \in \mathbb{Z}_+$
    % \begin{equation*}
    %     S(t+1) \leftarrow \underset{\sum S_i \leq C}{\max} \sum_{i=1}^N \nabla U_i (X_i(t)) \cdot \left( \frac{1 - \hat{\alpha}_i(t)^{S_i( t + 1)+ 1}}{1 - \hat{\alpha}_i(t)} \right) 
    % \end{equation*}
 %   \STATE The \emph{verification server} sends $S(t+1)$ to draft servers. (\hyperlink{step:6}{step \textcircled{6}})

%\input{Section/problemformulation/algorithmdesign}
\subsection{Convergence Analysis}
\label{sec:con_analysis}
The purpose of this section is to establish the asymptotic convergence of the \OurAlg's gradient scheduling to the optimal goodput  \(x^*\) of problem (\ref{primal}). Given the non-stationary nature of acceptance rates \(\hat{\alpha}_i(t)\) due to varying prompts, we adapt fluid sample path techniques to analyze the long-term behavior of the smoothed estimates \(X(t)\) and \(\hat{\alpha}(t)\). Under suitable assumptions, we show that the framework achieves optimality in the fluid limit as the smoothing parameter \(\beta \to 0\).

% \subsubsection{Assumptions}
% The following assumptions ensure the stability and convergence of the estimates despite non-stationarity.

\begin{assumption} [\textbf{Ergodicity with Converging Time Average}]
\label{assumption:1}
For each client \(i \in \{1, 2, \dots, N\}\), the sequence of acceptance rates \(\{{\hat{\alpha}}_i(t), t = 0, 1, 2, \dots\}\) is ergodic, such that the time average
\begin{equation}
\bar{\alpha}_i(T) = \frac{1}{T} \sum_{t=0}^{T-1} \hat{\alpha}_i(t)  \nonumber
\end{equation}
converges almost surely to a limit \(\bar{\alpha}_i \in [0, 1]\) as \(T \to \infty\).
% The assumption of ergodicity implies that, despite non-stationary prompt distributions, the long-term time average of \( \hat{\alpha}_i(t) \) stabilizes to a constant \( \bar{\hat{\alpha}}_i \). This is grounded in the update rule, which is a stochastic approximation process where \( \eta \) acts as a diminishing step size. 

% The ergodicity holds if the sequence of acceptance rates, driven by varying prompts, forms an ergodic process under a time-average measure. Theoretical support comes from stochastic approximation theory \cite{Approximation}, which states that for a process with a bounded update term (here, \( \min\left\{1, \frac{p{(s_j)}}{q_i(s_j)}\right\} \in [0, 1] \)) and \( \eta = O(1/t) \), the time average converges to the expected value of the driving noise, provided the process is mixing over time. 
This assumption is realistic in distributed edge-server setups, and the ergodicity requirement is met if prompt variations (e.g., topic shifts in real-time translation) are sufficiently mixing, which is typical in conversational AI where user inputs follow predictable patterns over time. The smoothing parameter \( \eta \) can be dynamically adjusted based on observed variance in $\hat{\alpha}_i(t)$ (e.g., reducing \( \eta \) if variance exceeds a threshold). The convergence of \( \bar{\alpha}_i(T) \) is empirically achievable within $700$ rounds, as detailed in Section~\ref{exp:conv_util_func} and Figure~\ref{fig:utility_conv}. 
\end{assumption}

\begin{assumption}[\textbf{Uniform Boundedness and Lipschitz Continuity}]
\label{assumption:2}
For all \(t \geq 0\) and each client \(i \in \{1, 2, \dots, N\}\), the acceptance rate \(\hat{\alpha}_i(t)\) is uniformly bounded by a constant $\hat{\alpha}_{\max}$, i.e., \(\hat{\alpha}_i(t) \leq \hat{\alpha}_{\max} < 1, \forall i, t\). Additionally, \(\hat{\alpha}_i(t)\) is Lipschitz continuous with a constant \(L > 0\), such that
\begin{equation}
|\hat{\alpha}_i(t+1) - \hat{\alpha}_i(t)| \leq L \cdot \eta \nonumber
\end{equation}
where \(L\) depends on the maximum variation in the empirical acceptance indicators \(\min\left(1, \frac{p_j(s_j)}{q_{i,j}(s_j)}\right)\), and \(\eta \in (0,1)\) is the smoothing parameter for \(\hat{\alpha}_i(t)\).
% Empirical evidence from speculative decoding systems shows acceptance rates plateau below 1 due to model misalignment or prompt complexity. Experiments with distilled models (e.g., T5-small vs. T5-XXL) report $\hat{\alpha}_i$ from 0.7 to 0.9 \cite{Yaniv}, supporting $\hat{\alpha}_{\text{max}} < 1$. This holds across non-stationary prompts like evolving chats, and setting $\hat{\alpha}_{\text{max}}$ to the 95th percentile (e.g., 0.9) via initial calibration is a practical standard. The update changes $\hat{\alpha}_i(t)$ by $\eta \left( \frac{1}{S_i(t)} \sum_{j=1}^{S_i(t)} \min\left(1, \frac{p_j}{q_{i,j}}\right) - \hat{\alpha}_i(t) \right)$, where the increment is bounded by $\eta$. 

The Lipschitz constant  $L$ captures the maximum rate of change, influenced by the variability of $\min\left(1, \frac{p_j}{q_{i,j}} \right)$. With $\eta = O(1/t^{\hat{\alpha}})$  ($\hat{\alpha} > 0.5$), the step size shrinks, ensuring smooth transitions. Stochastic approximation theory \cite{Approximation} guarantees Lipschitz continuity for bounded noise terms, with $L \leq 1$ derived as the maximum gradient, making this assumption theoretically plausible.
\end{assumption}

\begin{assumption} [\textbf{Sufficient Smoothing and Step Size}]
\label{assumption:3}
The smoothing parameter for acceptance rates is chosen as  $\eta = O(1/t^{\hat{\alpha}})$ with $\hat{\alpha} \in (0.5, 1]$, 
and the step size for goodput estimates satisfies \(\beta = O(1/t^\beta)\) with \(\beta \in (0.5, 1]\), ensuring that the ratio \(\eta / \beta \to 0\) as \(t \to \infty\).
\end{assumption}

% \begin{theorem}
% \label{theorem:main}
% \textbf{(Uniform Convergence in Probability)} For bounded \(A\) and \(\epsilon > 0\), there exists \(T > 0\) such that
% \[
% \lim_{\beta \to 0} \sup_{X^\beta(0) \in A, t > T/\beta} \Pr\left( \|X^\beta(t) - x^*\| > \epsilon \right) = 0.
% \]
% \end{theorem}
% \begin{proof}
% See Appendix. 
% \end{proof}

\begin{theorem}[\textbf{Uniform Convergence in Probability}]
\label{theorem:main}
Let $\{X^\beta(t)\}_{t \geq 0}$ be the scaled stochastic process governed by the proposed algorithm with step-size parameter $\beta > 0$. Suppose the initial condition $X^\beta(0)$ lies within a bounded set $A \subset \mathbb{R}_+^N$. Under Assumptions~\ref{assumption:1}–\ref{assumption:3}, for any $\epsilon > 0$, there exists a time threshold $T > 0$ such that
$$
\lim_{\beta \to 0} \sup_{X^\beta(0) \in A, \, t > T/\beta} \Pr\left( \|X^\beta(t) - x^*\| > \epsilon \right) = 0,
$$
where $x^* \in \mathbb{R}_+^N$ is the unique globally attractive fixed point of the fluid dynamics.

\end{theorem}

\begin{proof}
The proof of this theorem leverages the following foundational theorems, which are formally established in Appendix~\ref{ap:proof}. First, Theorem~\ref{theorem:1} shows that any weak limit of the scaled process $X^\beta(t)$ is almost surely a Fluid Sample Path (FSP). Then, Theorem~\ref{theorem:2} ensures that all FSPs starting from bounded initial conditions converge uniformly to the globally attractive fixed point $x^*$, due to a Lyapunov argument and boundary drift. Finally, Theorem~\ref{theorem:3} guarantees that the time-averaged performance also converges in expectation to $x^*$. Together, these imply that the stochastic process concentrates near $x^*$ with high probability for sufficiently small $\beta$ and large enough time, yielding the stated uniform convergence in probability.
\end{proof}

\section{Performance Evaluation}
\subsection{Experimental Setting}

\subsubsection{Testbed} To evaluate \OurAlg, we conducted experiments in a distributed inference environment designed to reflect practical deployment scenarios, where a large verification server collaborates with multiple lightweight draft servers to handle heterogeneous workloads. The verification server, equipped with an NVIDIA H100 GPU, hosted two target models—Qwen3-14B and Llama3.1-70B-Instruct-AWQ-INT4—representing distinct verification configurations. Meanwhile, each draft server, running on an NVIDIA L4 GPU, performed initial generation using lightweight models from the Llama3 and Qwen3 families, with sizes ranging from 0.6B to 3B parameters. This heterogeneous setup enabled us to evaluate the framework's ability to optimize resource usage across diverse model sizes and inference demands. In each time slot, the verification server processed batched draft outputs collected from the draft servers. Table~\ref{tab:exp_configs} summarizes the experimental configurations, including model choices, hardware specifications, the number of draft servers (clients), and token budget constraints.

\subsubsection{Datasets}
The experiments leveraged eight public datasets to simulate real-world applications. \textit{Alpaca} and \textit{Awesome-ChatGPT-Prompts} focus on instruction tuning and conversational Q/A; \textit{CNN/DailyMail} targets long-context summarization; \textit{OpenOrca} and \textit{Chatbot Arena} involve reasoning and open-domain Q/A; \textit{GSM8K} addresses mathematical problem solving; \textit{SPIDER} covers text-to-SQL generation; and \textit{HLE} includes high-difficulty, long-tail queries. These datasets were distributed across four to eight draft servers to create a mix of short, interactive prompts and longer, compute-intensive tasks. Maximum token lengths were set to 50 or 150 tokens depending on the setup, reflecting typical constraints in latency-sensitive environments. To increase task diversity among clients, each draft server was assigned a distinct dataset from the set above.

\subsubsection{Hyperparameters Choice} A critical experimental parameter in our system is the capacity $C$. We select $C$ based on two primary factors: (1) \textit{H100 HBM3 memory usage}: Considering the memory footprint of the verification server models (\textit{Qwen3-14B}, \textit{Llama3.1-70B-Instruct-AWQ-INT4}), including model weights, intermediate buffers, and output logits, and assuming the maximum prefix length observed in our datasets, we constrain memory usage to remain below 75\% of the HBM3's total 80GB capacity to prevent fragmentation and ensure stability. (2) \textit{Latency tolerance}: A larger $C$ increases the speculative draft length, but also leads to longer transmission time between verification and draft server due to the need to carry the full probability distributions for more tokens. This can delay response time in latency-sensitive applications.

To balance available memory headroom and latency constraints, we set the SD budget to $C=\{6, 20\}$ for 150-token generations and $C=\{24, 28\}$ for 50-token generations, which helps achieve a stable trade-off between efficiency and system responsiveness.

\begin{table*}
  \centering
  \caption{Experimental configurations}
  \label{tab:skew_result}
  % resize makesure the table won't out of page
  \resizebox{0.95\linewidth}{!}{%
    \begin{tabular}{lcccccc}
      \toprule
       \textbf{Verification Model} & \textbf{Draft Model} & \textbf{Verification Server GPU} & \textbf{Draft Server GPU} & \textbf{$C$ (tokens)} & \textbf{Number of Draft Servers} & \textbf{Max Token Length} \\
      \midrule
      Qwen3-14B     & Qwen3-0.6B  & H100 & L4 &24, 28 & 4 & 50 \\
      \midrule
      Qwen3-14B     & Qwen3-0.6B/1.7B  & H100 & L4 &16, 20 & 8 & 150 \\
      \midrule
      Llama-3.1-70B-Instruct    & Llama 3.2-1B-Instruct/  & \multirow{2}{*}{H100} & \multirow{2}{*}{L4}& \multirow{2}{*}{16, 20} & \multirow{2}{*}{8} & \multirow{2}{*}{150} \\
      -AWQ-INT4 & Llama 3.2-3B-Instruct\\
      \bottomrule
    \end{tabular}%
  }
\label{tab:exp_configs}
\end{table*}
\subsection{Main Results}
\subsubsection{Goodput Estimation} We evaluate the effectiveness of our smoothed goodput estimation by comparing it with the actual system-level goodput over time. Figure~\ref{fig:goodput_estimation} shows results for two representative scenarios—Qwen3 and Llama3—each with 8 clients operating under dynamic and heterogeneous prompt conditions.
To reduce transient fluctuations, we apply a moving average (MA) filter with a window size of 10 to both the estimated and measured goodput curves. The plots demonstrate a strong alignment between our estimated and the ground-truth goodput, indicating that the proposed exponential smoothing method captures system dynamics despite the stochastic nature of SD and clients' prompt variability.  We visualize the empirical standard deviation (square root of MA variance) as shaded confidence bands around both curves. 

These regions encompass most observed goodput peaks, further confirming the stability and predictive fidelity of our estimation process. This supports our design choice to use smoothed estimates in the gradient scheduling algorithm, allowing for robust and adaptive token allocation across non-stationary environments and diverse LLMs. % diverse LLM configurations.

\begin{figure*}[t]
    \centering
    % \begin{subfigure}[b]{0.32\textwidth}
    %     \centering
    %     \includegraphics[width=\textwidth]{Figures/Experimental_Results/Real_Vs_Estimate_GoodPut/real_vs_estimate_goodput_qwen3_4clients.pdf}
    %     \caption{Qwen – 4 clients}
    %     \label{fig:goodput-qwen3_4}
    % \end{subfigure}
    % \begin{subfigure}[b]{0.32\textwidth}
    \begin{subfigure}[b]{0.48\textwidth}
        \centering
        \includegraphics[width=\textwidth]{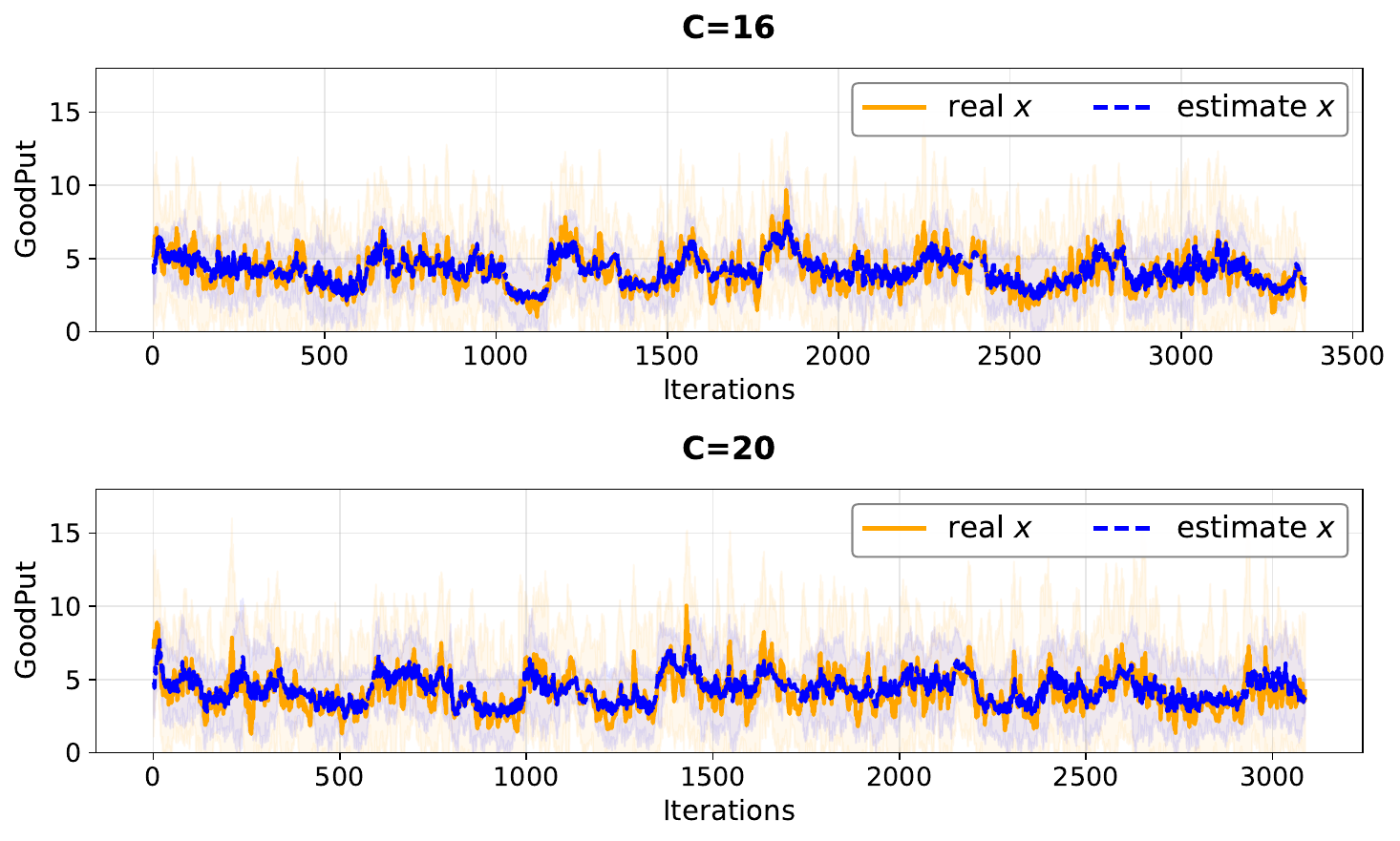}
        \caption{Qwen3 – 8 clients}
        \label{fig:goodput-qwen3}
    \end{subfigure}
    %\hfill
    % \begin{subfigure}[b]{0.32\textwidth}
    \begin{subfigure}[b]{0.48\textwidth}
        \centering
        \includegraphics[width=\textwidth]{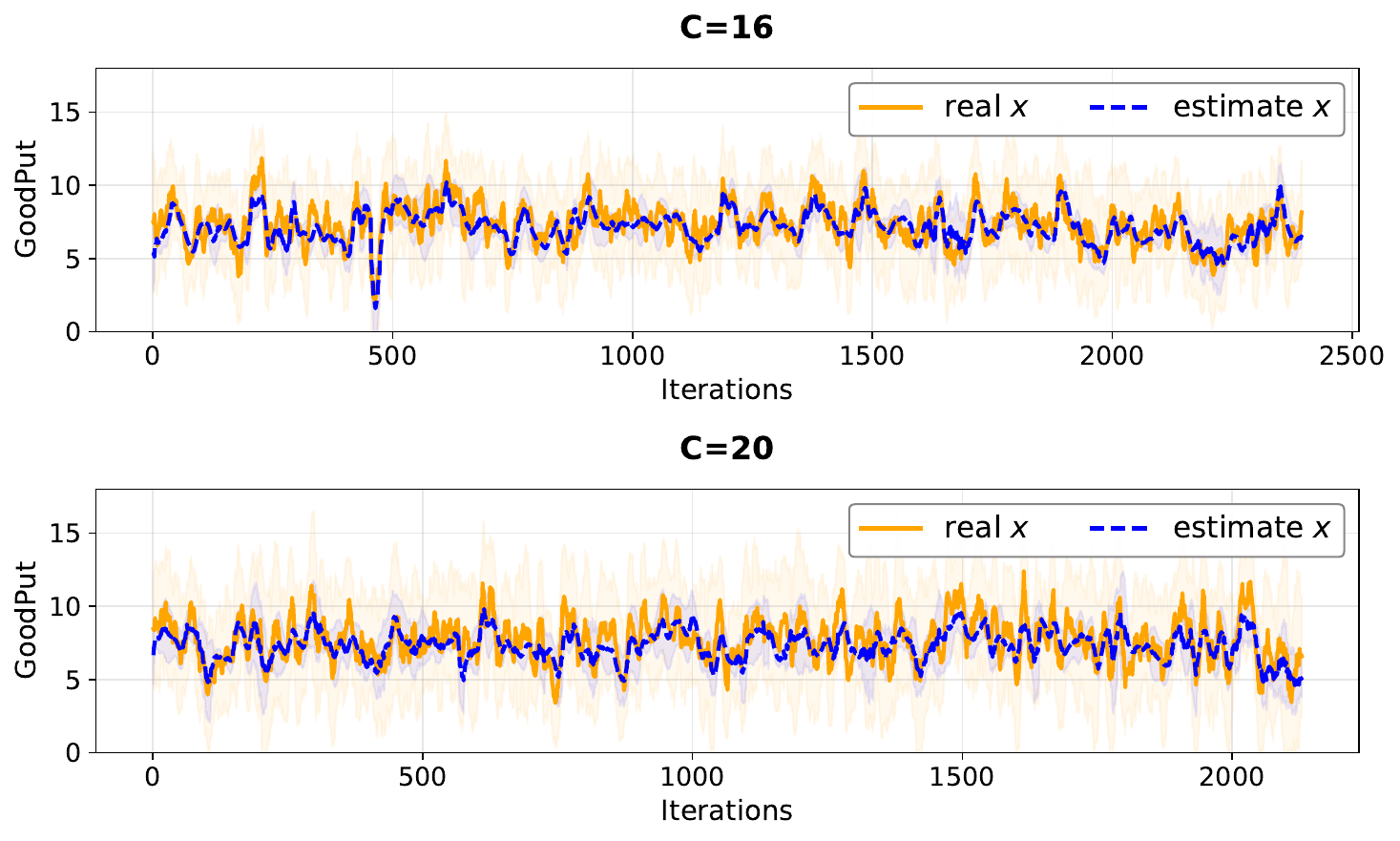}
        \caption{Llama3 – 8 clients}
        \label{fig:goodput-llama3}
    \end{subfigure}
    \caption{Comparison of estimated goodput vs real goodput across different LLMs with 8 clients.}
    \label{fig:goodput_estimation}
\end{figure*}
\subsubsection{Time Distribution} We evaluate the end-to-end wall time of the \textsc{\OurAlg} framework and compare it against two common SD baselines: (1) \textit{Fixed-S}: each draft server consistently generates a fixed number of tokens per iteration, with $S_i = C / N$; and (2) \textit{Random-S}: each draft server randomly samples $S_i$ per iteration, constrained such that the total across clients does not exceed $C$. In Fig.~\ref{fig:time-distribution}, we see that \OurAlg achieves comparable total wall time to the Fixed-S baseline across both Qwen3 and Llama3 settings. In contrast, Random-S exhibits a 5--25\% increase in wall time due to scheduling inefficiencies. To better understand this behavior, we decompose wall time into three components: (1) \textit{Receiving time}, the duration the verification server waits for draft servers to complete token generation and transmit their draft distributions, until the batch is fully assembled; (2) \textit{Verification time}, the time taken by the verification process; 

and (3) \textit{Sending time}, the duration for communicating the accepted tokens and next-round scheduling instructions back to draft servers.\\

We see that the receiving and verification dominate total wall time, while sending time contributes less than 0.1\%. The increased receiving time in \OurAlg and Random-S results from variable drafting lengths, requiring the verification server to wait for the slowest client (largest $S_i$), unlike Fixed-S, which enables faster batch formation without this bottleneck. In contrast, \OurAlg achieves ~5\% lower verification time than Fixed-S, reflecting more efficient workload balancing under dynamic client behavior.
\vspace{-1mm}
\subsubsection{Convergence of Utility Function}
Figure \ref{fig:utility_conv} illustrates the utility function $ U(\cdot) $ computed on the empirical average goodput $\bar{x}(T) = \frac{1}{T} \sum_{t=1}^{T} x(t)$ over the first 600 iterations, showing that the \OurAlg curve starts lower due to initial exploration but rises steadily, stabilizing by approximately iteration 400, and consistently surpasses the utilities of the baselines (``Fixed-$S$" for uniform allocation and ``Random-$S$" for stochastic assignment).

The observed stabilization of $ U(\bar{x}(T)) $ within 600 iterations shows that the system is converging, which can be explained by looking at the utility function and how the system works. The utility function $ U(x) = \sum_{i=1}^N \log x_i $ is smooth and achieves the highest value at the optimal $ x^* $. As the smoothed goodput estimate $ X^\beta(t) $ gets closer to $ x^* $ over time (c.f. Theorem \ref{theorem:2}), the value of $ U(X^\beta(t)) $ rises and levels off toward $ U(x^*) $. This smoothed estimate $ X^\beta(t) $ helps decide how many speculative tokens $ S_i(t) $ each client gets, using a simple adjustment method, which moves the system toward the optimal allocation. Therefore, the observed stabilization of $ U(\bar{x}(T)) $ within 600 iterations suggests that the utility reaching a steady value close to $ U(x^*) $ after the initial adjustments, matching Theorem \ref{theorem:main}. With $\beta = 0.5$, the threshold $ T/\beta \approx 400-600 $ iterations (derived from $ T \approx 200-300 $ rounds) aligns well with the stabilization range, indicating the system surpasses the point where the probability of significant deviation $\|X^\beta(t) - x^*\| > \epsilon$ becomes negligible, a trend supported by the lack of oscillations after 400 iterations. 

\begin{figure*}[t]
    \centering
    \begin{subfigure}[b]{0.96\textwidth}
        \centering
        \includegraphics[width=\textwidth]{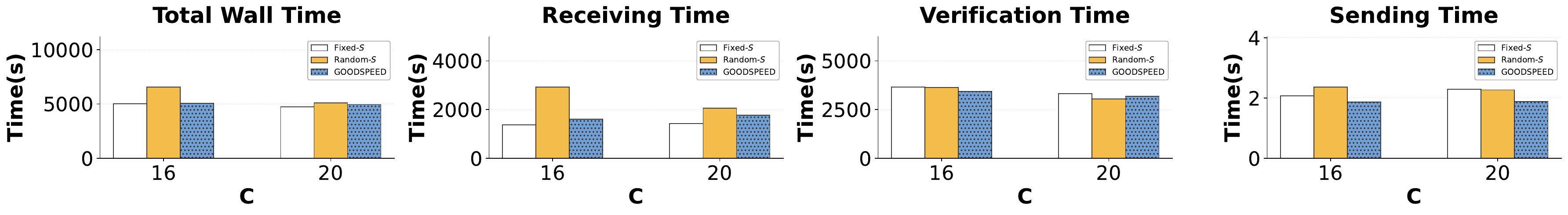}
        \caption{Qwen3 – 8 clients}
        \label{fig:time-dist-qwen3}
    \end{subfigure}
    \hfill
    \begin{subfigure}[b]{0.96\textwidth}
        \centering
        \includegraphics[width=\textwidth]{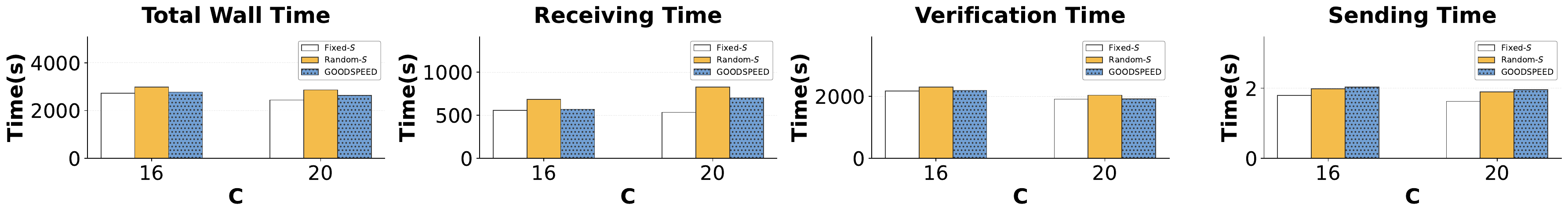}
        \caption{Llama3 – 8 clients}
        \label{fig:time-dist-llama3}
    \end{subfigure}
    \caption{Comparison of time distribution across different LLMs with 8 clients.}
    \label{fig:time-distribution}
\end{figure*}

\label{exp:conv_util_func}
\begin{figure*}[t]
    \centering
    \begin{subfigure}[b]{0.32\textwidth}
        \centering
        \includegraphics[width=\textwidth]{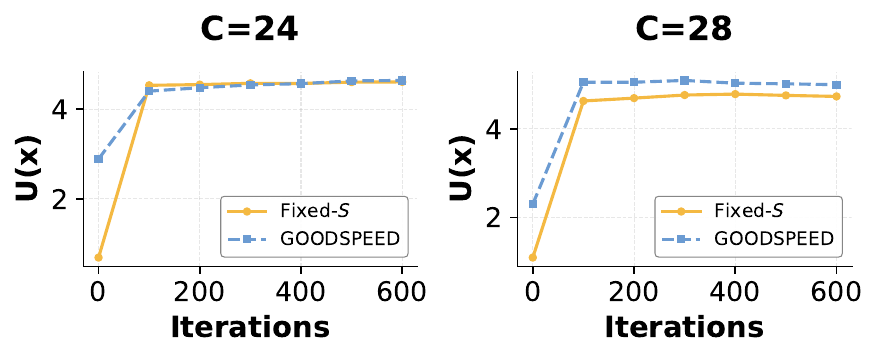}
        \caption{Qwen3 – 4 clients}
        \label{fig:utility-qwen3-4}
    \end{subfigure}
    %\hfill
    \begin{subfigure}[b]{0.33\textwidth}
        \centering
        \includegraphics[width=\textwidth]{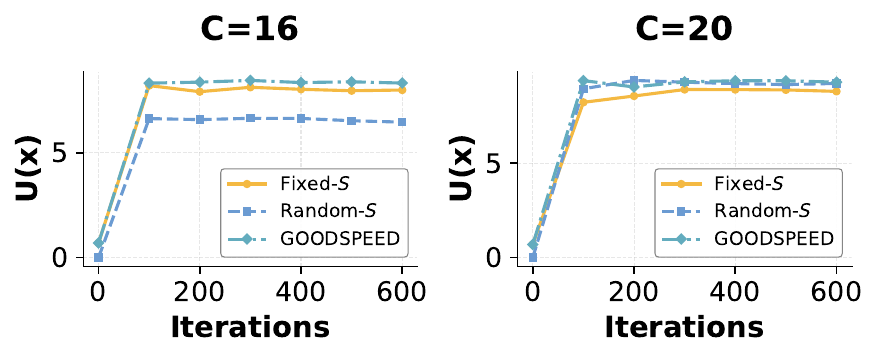}
        \caption{Qwen3 – 8 clients}
        \label{fig:utility-qwen3-8}
    \end{subfigure}
    %\hfill
    \begin{subfigure}[b]{0.33\textwidth}
        \centering
        \includegraphics[width=\textwidth]{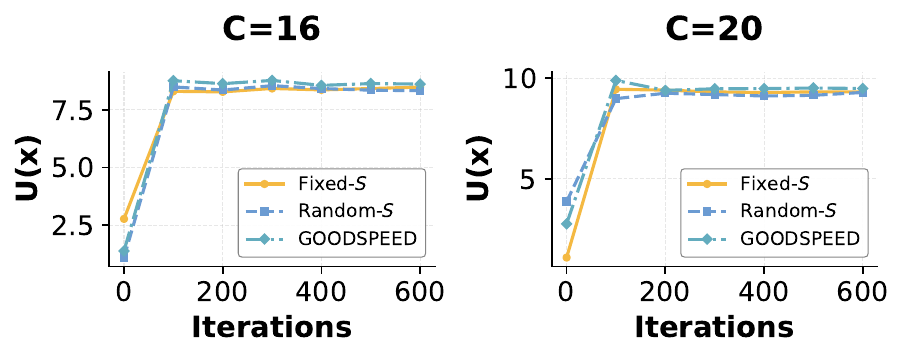}
        \caption{Llama3 – 8 clients}
        \label{fig:utility-llama3-8}
    \end{subfigure}

    \caption{Convergence of utility function $U$ over iterations for different models and client settings.}
    \label{fig:utility_conv}
\end{figure*}

\section{Conclusion}

We propose \OurAlg, a distributed LLM inference framework that combines speculative decoding with gradient-based scheduling for efficient and adaptive resource allocation between a verification server and multiple draft servers. By formulating inference as a utility maximization problem over the achievable goodput region, we define a clear performance benchmark. Our scheduling algorithm approximates the optimal allocation in real time, adapting to non-stationary acceptance rates. Experiments demonstrate fast convergence within 600 iterations and superior efficiency and fairness compared to baselines across diverse datasets and model configurations.
%Future work could explore dynamic capacity adjustments or multi-edge architectures to further enhance scalability and robustness in real-world deployments.
%verall, GOODSPEED advances fair and efficient LLM serving, offering promising avenues for edge computing and interactive AI systems.
\section{Appendix}
\label{ap:proof}
In this section, we establish a sketch of proof for the asymptotic convergence of the \OurAlg framework to the optimal goodput  $x^*$ from the problem \eqref{primal}. 
We begin by producing key lemmas on fluid sample path (inspired from \cite{Stolyar}) properties, and finally prove the main theorem on process convergence, uniform attraction, and asymptotic optimality. We consider a sequence of systems indexed by $\beta \rightarrow 0$ along $\mathscr{B} = \{\beta_j, j = 1, 2, \dots\}$ with $\beta_j > 0$. Processes are superscripted by $\beta$, e.g., $X^\beta(t)$, $\alpha^\beta(t)$, $S^\beta(t)$. The fluid-scaled processes are defined as follows: 
\begin{itemize}
    \item Scaled goodput estimates: $x^\beta(t) = X^\beta(t / \beta)$, $t \in \mathbb{R}_{+}$
    \item Scaled acceptance rates: $\tilde{\alpha}^\beta(t) = \alpha^\beta(t / \beta)$
    \item Scaled realized goodputs: $\hat{f}^\beta(t) = \beta \sum_{l=1}^{\lfloor t / \beta \rfloor} x^\beta(l)$, where $x^\beta(l)$ is the realized goodput vector at step $l$
    \item Scaled allocation counts: $\hat{g}_{k}^\beta(t) = \beta \hat{G}_{k}^\beta(t / \beta)$, where $\hat{G}_{k}^\beta(\tau)$ counts steps up to $\tau$ using decision $k \in \mathcal{K}$`
\end{itemize}

The composite scaled process is $z^\beta = (x^\beta, \tilde{\alpha}^\beta, \hat{f}^\beta, \hat{g}^\beta)$. A fixed set of functions $z = (x, \tilde{\alpha}, \hat{f}, \hat{g})$ is defined as a fluid sample path (FSP) if there exists a subsequence $\mathscr{B}_0 \subseteq \mathscr{B}$ such that $z^\beta \to z$ uniformly on compact sets (u.o.c.) as $\beta \to 0$ along $\mathscr{B}_0$, with the initial condition satisfying $\|x(0)\| < \infty$.

%%----------------------Lemmas-------------
% \newtheorem{lemma}{Lemma}
\begin{lemma}
\label{lemma:1}
\textbf{(Basic Properties of FSPs)}  
For any FSP $z$, all component functions are Lipschitz continuous on $[0, \infty)$ with Lipschitz constant $C + \|x(0)\|$, where $C > 0$ is a constant depending only on the system parameters (including bounds $\underline{\mu}$ and $\bar{\mu}$). Moreover, the functions $\hat{f}(t)$ and $\hat{g}_{k}(t)$ are nondecreasing and satisfy the relations  
\begin{align*}
\hat{f}(t) = \sum\nolimits_{k \in \mathcal{K}} \mu(k; \tilde{\alpha}(t)) \hat{g}_{k}(t), \quad t \geq 0 
\end{align*}
where $\mu(k; \tilde{\alpha}(t))$ is the goodput vector evaluated at the scaled acceptance rates $\tilde{\alpha}(t)$, and the total time measure is normalized as $\sum_{k \in \mathcal{K}} \hat{g}_{k}(t) = t$ for $t \geq 0$.
\end{lemma}

%---proof----------------
\begin{proof}
% The update $ X_n^\beta(l) = (1 - \beta) X_n^\beta(l-1) + \beta x_n^\beta(l)$ yields  
% \[
% |X_n^\beta(l) - X_n^\beta(l-1)| \leq \beta (2\bar{\mu} + X_n^\beta(0)),
% \]  
% since $ x_n^\beta(l) \leq \bar{\mu} $. 

% The acceptance update $ \alpha_n^\beta(l) = (1 - \eta) \alpha_n^\beta(l-1) + \eta \cdot \frac{1}{S_n^\beta(l)} \sum_{j=1}^{S_n^\beta(l)} \min\left(1, \frac{p_{n,j}}{q_{n,j}}\right) $ has differences bounded by $\eta$. Scaling by $\beta$ and taking u.o.c. limits gives Lipschitz continuity with $ C = 2\bar{\mu} + L $. Nondecreasing properties follow from $\hat{f}^\beta(t)$ and $\hat{g}_{k}^\beta(t)$ as cumulative sums, and the relations hold via $\lim_{\beta \to 0} \beta \sum_{l} = \int \frac{d}{dt} \hat{f}(t) dt$ and Assumption \ref{assumption:1}.
The discrete update for goodput estimates is given by $ X_n^\beta(l) = (1 - \beta) X_n^\beta(l-1) + \beta x_n^\beta(l) $, where $ x_n^\beta(l) $ is the realized goodput at step $ l $, bounded by $\bar{\mu}$ due to the verification constraint. The difference between consecutive updates is
\[
|X_n^\beta(l) - X_n^\beta(l-1)| \leq \beta (2\bar{\mu} + X_n^\beta(0)),
\]
since $ x_n^\beta(l) \leq \bar{\mu} $ and the initial condition $ X_n^\beta(0) $ is finite. For acceptance rates, the update rule \ref{updatealpha}
% $\alpha_n^\beta(l) = (1 - \eta) \alpha_n^\beta(l-1) + \eta  \frac{1}{S_n^\beta(l)} \sum_{j=1}^{S_n^\beta(l)} \min\left(1, \frac{p_j(x_j)}{q_{n,j}(x_j)}\right)$ 
results in a change bounded by $\eta$ times the maximum variation of the acceptance indicator, which is at most 1. Scaling these updates by $\beta$ and taking the uniform on compact sets (u.o.c.) limit as $\beta \to 0$ along $\mathscr{B}_0$, the Lipschitz constant of the limit functions is determined as $ C = 2\bar{\mu} + L $, where $ L $ arises from the Lipschitz continuity of $\alpha_i(t)$ as per Assumption \ref{assumption:2}. The nondecreasing property of $\hat{f}^\beta(t)$ and $\hat{g}_{k}^\beta(t)$ stems from their definition as cumulative sums over discrete steps. In the limit, $\hat{f}(t) = \lim_{\beta \to 0} \hat{f}^\beta(t)$ and $\hat{g}_{k}(t) = \lim_{\beta \to 0} \hat{g}_{k}^\beta(t)$ retain this monotonicity. The relation $\hat{f}(t) = \sum_{k \in \mathcal{K}} \mu(k; \tilde{\alpha}(t)) \hat{g}_{k}(t)$ holds because the realized goodput is the accepted portion of allocated tokens, with $\mu(k; \tilde{\alpha}(t))$ representing the expected goodput for allocation $ k $ given $\tilde{\alpha}(t)$. 
% The normalization $\sum_{k \in \mathcal{K}} \hat{g}_{k}(t) = t$ is enforced by the scheduling algorithm's capacity constraint under fixed $ C $, as validated by Assumption \ref{assumption:1}'s ergodicity.
\end{proof}
\begin{lemma}
\label{lemma:2}
\textbf{(Differential Inclusion and Boundary Conditions)}  
For almost all $ t \geq 0 $, an FSP $ z $ satisfies  
\[
x'(t) = v(t) - x(t),
\]    
\[
v(t) = \sum\nolimits_{k \in \mathcal{K}} \hat{g}_{k}'(t) \mu(k; \tilde{\alpha}(t)) \in \arg\max_{v \in \mathcal{X}(t)} \sum\nolimits_{i=1}^N \frac{1}{x_i(t)} v_i,
\]  
and $\mathcal{X}(t)$ uses $\mu(k; \tilde{\alpha}(t))$. If $ x_B(t) = 0 $ for $ B \subseteq \{1, \dots, N\} $ and $ x_i(t) > 0 $ for $ i \notin B $, then  with $ c \geq \underline{\mu} $, we have $\frac{d_+}{dt} \sum_{i \in B} x_i(t) \geq c > 0$.
\end{lemma}

%-----------proof------------------
% The integral equation $ x_n(t) - x_n(0) = \hat{f}_n(t) - \int_0^t x_n(\xi) d\xi $ differentiates to $ x'(t) = v(t) - x(t) $ at regular points, where $ v(t) = \frac{d}{dt} \hat{f}(t) $. For $ v(t) $, $\nabla U_i(x_i) = \frac{1}{x_i}$ for $ U(x) = \sum \log x_i $, so $ v(t) \in \arg\max_{v \in \mathcal{X}(t)} \sum_{i=1}^N \frac{1}{x_i(t)} v_i $, reflecting the logarithmic utility’s fairness focus. At $ x_B(t) = 0 $, $\frac{1}{x_i(t)} \to +\infty$, prioritizing $ B $. With $\underline{\mu} > 0$, an allocation $ k \in \mathcal{K} $ to $ B $ gives $\sum_{i \in B} \mu_i(k; \tilde{\alpha}(t)) \geq \underline{\mu}$, so $\frac{d_+}{dt} \sum_{i \in B} x_i(t) \geq \underline{\mu} > 0$.
\begin{proof}
The differential equation is derived from the integral form of the FSP. The scaled goodput update leads to the relation $ x_n(t) - x_n(0) = \hat{f}_n(t) - \int_0^t x_n(\xi) d\xi $, where $\hat{f}_n(t)$ is the cumulative realized goodput. Differentiating with respect to $ t $ at points of continuity (almost everywhere due to Lipschitz continuity from Lemma \ref{lemma:1}), we obtain
 $x'(t) = \frac{d}{dt} \hat{f}(t) - x(t)$. Since $\hat{f}(t) = \sum_{k \in \mathcal{K}} \mu(k; \tilde{\alpha}(t)) \hat{g}_{k}(t)$ and $\hat{g}_{k}(t)$ is absolutely continuous as a limit of step functions, the derivative is $\frac{d}{dt} \hat{f}(t) = \sum_{k \in \mathcal{K}} \hat{g}_{k}'(t) \mu(k; \tilde{\alpha}(t))$. The scheduling algorithm selects $\hat{g}_{k}'(t)$ to maximize $\nabla U(x(t)) \cdot v$, where $\nabla U_i(x_i) = \frac{1}{x_i}$ for $ U(x) = \sum_i \log x_i $. Thus,
\[
v(t) = \sum\nolimits_{k \in \mathcal{K}} \hat{g}_{k}'(t) \mu(k; \tilde{\alpha}(t)) \in \arg\max_{v \in \mathcal{X}(t)} \sum\nolimits_{i=1}^N \frac{1}{x_i(t)} v_i,
\]
reflecting the logarithmic utility’s emphasis on fairness by prioritizing clients with lower $ x_i(t) $. For the boundary condition, if $ x_B(t) = 0 $ for some $ B $, the gradient $\frac{1}{x_i(t)} \to +\infty$ as $ x_i(t) \to 0^+ $, compelling the maximization to allocate tokens to $ B $. Since $\mathcal{X}(t)$ includes points $\mu(k; \tilde{\alpha}(t))$ with $\sum_{i \in B} \mu_i(k; \tilde{\alpha}(t)) \geq \underline{\mu}$ for some $ k $ (due to $\underline{\mu} > 0$), and $\sum_{k} \hat{g}_{k}'(t) = 1$ with $\hat{g}_{k}'(t) \geq 0$, we have $\frac{d_+}{dt} \sum_{i \in B} x_i(t)$
\[
 = \sum_{i \in B} \sum_{k \in \mathcal{K}} \hat{g}_{k}'(t) \mu_i(k; \tilde{\alpha}(t)) \geq \underline{\mu} \sum_{k: \mu_B(k) \geq \underline{\mu}} \hat{g}_{k}'(t) \geq \underline{\mu} > 0\]
where $\mu_B(k) = \sum_{i \in B} \mu_i(k; \tilde{\alpha}(t))$, ensuring a positive right-hand derivative.
\end{proof}
%-----------------------------------------------
\begin{theorem}
\label{theorem:1}
\textbf{(Process-Level Convergence):} The sequence $\{z^\beta\}$ is relatively compact, and any weak limit is concentrated on FSPs with probability 1.
\end{theorem}
\begin{proof}
Relative compactness holds as $ |z^\beta(t_2) - z^\beta(t_1)| \leq C^*(t_2 - t_1) + \epsilon $ in probability, from Lipschitz bounds and $\bar{\mu}$, $\alpha_{\max}$. Using Skorohod representation, converging subsequences are almost surely FSPs, as scaled updates and Assumption \ref{assumption:1} ensure consistency.
\end{proof}
\begin{theorem}
\label{theorem:2}
\textbf{(Uniform Attraction of Fluid Sample Paths)} For any bounded set $A \subset \mathbb{R}_{+}^N$, FSPs satisfy $x(t) \to x^*$ as $t \to \infty$, uniformly on $x(0) \in A$.
\end{theorem}
\begin{proof}
The distance $\rho(x(t), \mathcal{X}(t)) \leq \rho(x(0), \mathcal{X}(0)) e^{-t}$ decays, as $ x'(t) = v(t) - x(t) $ with $ v(t) \in \mathcal{X}(t) $. Therefore, if $ x_B(t) = 0 $, the inequality $\frac{d_+}{dt} \sum_{i \in B} x_i(t) \geq \underline{\mu} > 0$ ensures $ x_i(t) > 0 $ after $ T_1 $, uniformly on $ A $. And outside $ O_\delta(x^*) $, $\sum_{i=1}^N \frac{1}{x_i(t)} (x^*_i - x_i(t)) > 0$, so $\frac{d}{dt} U(x(t)) > 0$, driving $ U(x(t)) \to U(x^*) $. Uniform convergence holds by finite $ T_2 $.
\end{proof}
\begin{theorem}
\label{theorem:3}
\textbf{(Asymptotic Optimality):} For bounded $A \subset \mathbb{R}_{+}^N$ and $\epsilon > 0$, there exist $T, T^* > 0$ such that
\[
\lim_{\beta \to 0} \sup_{X^\beta(0) \in A, l_1 > T/\beta, l_2 - l_1 > T^*/\beta} \left\| \mathbb{E} U^\beta(l_1, l_2) - x^* \right\| < \epsilon,
\]
where $U^\beta(l_1, l_2)$ is the average realized goodput over $[l_1, l_2]$.
\end{theorem}
\begin{proof}
Theorems \ref{theorem:1} and \ref{theorem:2} ensure $x^\beta(t) \to x^*$. Assumption \ref{assumption:3}’s $\frac{\eta}{\beta} \to 0$ stabilizes $\mathcal{X}(t) \to \mathcal{X}$, so interval averages converge uniformly in expectation.
\end{proof}
\textbf{Corollary 1:} In the stationary regime, $\lim_{\beta \to 0} \mathbb{E} x^\beta(1) = x^*$.

%------------------Removed theorem 4 (already explained in the paper)-------------

% \begin{theorem}
% \label{theorem:4}
% \textbf{(Uniform Convergence in Probability)} For bounded $A$ and $\epsilon > 0$, there exists $T > 0$ such that
% \[
% \lim_{\beta \to 0} \sup_{X^\beta(0) \in A, t > T/\beta} \Pr\left( \|X^\beta(t) - x^*\| > \epsilon \right) = 0.
% \]
% \end{theorem}

% \begin{proof}
% From Theorem \ref{theorem:1}, weak limits are FSPs. Theorem \ref{theorem:2}’s uniform attraction, with boundary drift, ensures $\|x(t) - x^*\| \to 0$. The probability of deviation diminishes as $ t > T/\beta $, with $T$ depending on $A$ and $\underline{\mu}$.
% \end{proof}

\bibliographystyle{IEEEtran}
\bibliography{references} % Path to your References.bib file
\end{document}